\documentclass[lettersize,onecolumn]{IEEEtran}
\usepackage{amsfonts,color,morefloats,pslatex,a4wide}
\usepackage{amssymb,amsthm,amsmath,latexsym,pslatex,cite}
\usepackage{lineno,hyperref,mathtools}
\usepackage{pdflscape}
\usepackage{float}
\usepackage{caption}
\usepackage{algorithm,algpseudocode}

\newtheorem{theorem}{Theorem}[section]
\newtheorem{notation}{Notation}
\newtheorem{lemma}[theorem]{Lemma}
\newtheorem{corollary}[theorem]{Corollary}
\newtheorem{proposition}[theorem]{Proposition}

\theoremstyle{definition}
\newtheorem{definition}[theorem]{Definition}
\newtheorem{example}[theorem]{Example}

\theoremstyle{remark}
\newtheorem{remark}[theorem]{Remark}

\newcommand{\F}{\mathbb{F}}
\newcommand{\C}{\mathcal{C}}
\newcommand{\bu}{\textbf{u}}
\newcommand{\bv}{\textbf{v}}

\begin{document}

\title{Shortest Embeddings of Linear Codes with Arbitrary Hull Dimension \thanks{The authors are with School of Mathematics and Statistics \& Hubei Key Laboratory of Mathematical Sciences, Central China Normal University, Wuhan China 430079.}
}

\author{ Jiabin Wang\qquad Jinquan Luo
\thanks{This research is supported by National Natural Science Foundation of China
		(Nos. 12441102, 12171191, 12271199 ), SRMC Fund with grant no. 2024SRMC01
		and the Fundamental Research Funds for the Central Universities with grant no. CCNU25JCPT031.}
\thanks{E-mails: wangjiabin0821@163.com(J.Wang), luojinquan@ccnu.edu.cn(J.Luo)}}

\markboth{Journal of \LaTeX\ Class Files,~Vol.~1, No.~4, April~2026}%
{Shell \MakeLowercase{\textit{et al.}}: A Sample Article Using IEEEtran.cls for IEEE Journals}


\maketitle

\begin{abstract}

In this paper, we study the shortest $t$-dimensional hull embeddings of linear codes in both Euclidean and Hermitian cases, extending the existing research on the shortest LCD and self-orthogonal embeddings to arbitrary hull dimensions and arbitrary finite fields.
		We obtain the exact length of such embeddings by adopting tools from quadratic form theory over finite fields and classical group theory. Based on the congruence equivalence class of Gram matrices of linear codes, we classify linear codes into distinct ``types'' and present corresponding constructive algorithms. In particular, we improve the results of An et al. \cite{An2025} and fully determine the length of the shortest self-orthogonal embeddings for linear codes. Finally, applying these algorithms, we provide examples for various settings and obtain several optimal codes inequivalent to those in the BKLC database.

\end{abstract}


\section{Introduction}
	The hull of a linear code, defined as the intersection of the code and its dual \cite{Assmus1990}, is a core concept in coding theory, as its dimension characterizes the structural properties of linear codes. Specifically, codes with a trivial hull (dimension $0$), known as linear complementary dual (LCD) codes, are introduced by Massey in solving  two-user binary adder channel problem \cite{Massey1992}. These codes later are found to have crucial applications in cryptography as countermeasures against side-channel and fault-injection attacks \cite{Carlet2016,Hossain2021}. On the other hand, codes with full hull (dimension equal to the code dimension), namely self-orthogonal codes, have long been fundamental in quantum information theory, serving as the key ingredient for the construction of quantum error-correcting codes (see \cite{Cald1998,Dast2024,Fred2025,Wang2024}). There have also been studies and applications of linear codes with intermediate hull dimensions\cite{Chen2023,Sok2022}.
	
	As a foundational work, the enumeration and asymptotic behavior of linear codes with a prescribed hull dimension have been thoroughly investigated. Luerssen and Ravagnani \cite{Glues2024} studied $\ell$-dimensional Euclidean hulls in finite bilinear spaces, derived enumeration and weight distribution results, and showed that random self-orthogonal codes are asymptotically MDS. Li et al. established closed mass formulas for linear codes with prescribed Euclidean hull dimension and classify small ternary codes \cite{Li2024}, further extending the result to Hermitian and symplectic hulls \cite{Li2025}.
	
	Recently, the embedding problem has attracted increasing attention: the core question is how to embed an arbitrary linear code into a target code with prescribed hull properties by adding the minimum number of coordinates. Specifically, Kim et al. \cite{Kim2021} proposed an algorithmic to embed a given binary $k$-dimensional linear code $\C$ $(k = 3, 4)$ into a self-orthogonal code of the shortest length that has preserves the dimension $k$ and satisfies $d \geq d(\C)$. For $k > 4$, they suggested a recursive method to embed a $k$-dimensional linear code into a self-orthogonal code. Furthermore, An et al. studied the shortest self-orthogonal embeddings of binary linear codes, proposed algorithms to construct self-dual codes from Hamming codes, and built several codes with new parameters \cite{An2025}. They later extended the framework to LCD codes, characterizing all forms of the shortest LCD embeddings of a linear code over small finite fields and constructing several new optimal LCD codes \cite{An2026}.
	
	However, most existing works only considers extreme hull dimensions-- LCD codes and self-orthogonal codes. No systematic research exists on the shortest embeddings for arbitrary hull dimensions and arbitrary finite fields. In fact, linear codes with intermediate hull dimensions are practically relevant, as they can be applied to construct entanglement-assisted quantum error-correcting codes (EAQECCs) (see \cite{Fang2020,Sok2022}), which motivates our work.
	
	In this paper, we extend existing research on the shortest LCD and self-orthogonal embeddings to arbitrary hull dimensions and arbitrary finite fields. Our results apply to both Euclidean and Hermitian cases. Our main goal is to determine the length of the shortest $t$-dimensional hull embeddings and provide explicit construction algorithms for such embeddings. For a linear code $\C$ with a generator matrix $G$, since the hull dimension of $\C$ is closely related to the rank of $GG^*$ (Theorem \ref{rahu}), our underlying idea is as follows: we first transform $GG^*$ into an equivalent ``canonical'' form, then adjust its rank by padding or eliminating entries at specific positions. Specifically, the main methods and results of this paper are summarized below.
	\begin{itemize}
		\item [(1)] In the Hermitian case, based on the congruence canonical form of Hermitian matrices, we prove that the length of the shortest $t$-dimensional Hermitian hull embeddings is $n+|t-\ell|$. In Section 5, we present an explicit algorithm for constructing such embeddings, and obtain an almost optimal $[11,6,4]_9$ code and an optimal $[6,3,4]_4$ Hermitian self-orthogonal code using this algorithm.
		\item [(2)] In the Euclidean case, we classify linear codes into four types based on the congruence equivalence classes of their Gram matrices. In Sections 3 and 4, we determine the exact lengths of the shortest $t$-dimensional Euclidean hull embeddings for each type through separate discussions. In Section 5, we present algorithms for constructing all such embeddings and apply these algorithms to obtain examples of good linear codes. In particular, we obtain several optimal self-orthogonal codes with parameters $[19,8,8]_3$, $[20,8,9]_3$, $[8,4,4]_2$, and $[20,6,8]_2$.
	\end{itemize}
	
	\section{Preliminaries}
	We begin with basic notions and fundamental results in coding theory \cite{Huff2003}.
	
	Let $\F_q$ be the finite field with $q$ elements. A $k$-dimensional subspace $\C$ of $\F_q^n$ is called a \textit{linear code} of length $n$ and
	dimension $k$, denoted as an $[n,k]$ code. The vectors of $\C$ are called \textit{codewords}. A \textit{generator matrix} of $\C$, written as
	$G(\C)$, is a $k \times n$ matrix whose rows form a basis of $\C$. Throughout this paper, $O$ denotes the all-zero matrices, $\textbf{0}$ denotes the all-zero vectors, and $I_j$ denotes the $j \times j$ identity matrix.
	
	For any $\bu=(u_1,\dots,u_n)\text{ and } \bv=(v_1,\dots,v_n) \in \F_q^n$, the \textit{Euclidean inner product} is defined as
	$$
	\left \langle \bu,\bv\right \rangle_E=\sum_{i=1}^{n}u_iv_i.
	$$
	 The \textit{Euclidean dual} of an $[n,k]$ code $\C$ over $\F_q$ is
	$$
	\C^{\bot_E}=\{\bv\in \F^n_q : \left \langle \bu,\bv\right \rangle_E=0 \text{ for all } \bu \in \C\}.
	$$
	Similarly, for $\bu,\bv\in \F_{q}^n$ (where $q$ is an even power of a prime), the \textit{Hermitian inner product} is
	$$
	\left \langle \bu,\bv\right \rangle_H=\sum_{i=1}^{n}u_iv^{\sqrt{q}}_i.
	$$
	The \textit{Hermitian dual} of $\C$ is
	$$
	\C^{\bot_H}=\{\bv\in \F^n_{q} : \left \langle \bu,\bv\right \rangle_H=0 \text{ for all } \bu \in \C\}.
	$$
	To deal with the Euclidean and Hermitian cases in a unified way, we next use the notation $\C^{\bot}$ to denote either the Euclidean dual or Hermitian dual of $\C$. The \emph{(Euclidean or Hermitian) hull} of a code $\C$ is defined as
	$$
	\text{Hull}(\C)=\C \cap \C^{\bot}.
	$$
	An $[n,k]$ code $\C$ is \textit{self-orthogonal} if $\dim (\text{Hull}(\C)) =k$, \textit{self-dual} if $\C = \C^{\bot}$, and \textit{linear complementary dual (LCD)} if $\dim (\text{Hull}(\C)) =0$.
	
	For any vector $\bu \in \F_q^n$, its \emph{Hamming weight} $\mathrm{wt}(\mathbf{u})$ is the number of nonzero components of $\bu$. For $\bu, \bv \in \F_q^n$, the \emph{Hamming distance} between them is defined as $d(\bu,\bv) = \mathrm{wt}(\bu-\bv)$. The \emph{minimum distance} of a linear code $\C$ is defined as
	$$d(\C)=\min \{\text{wt}(\bu) : \bu\in \C\setminus \{\textbf{0}\}\}.$$
	An $[n,k,d]_q$ code is a $q$-ary $[n,k]$ linear code with minimum distance $d$. An $[n,k,d]_q$ linear code is called \emph{ (distance) optimal} if there is no  $[n, k, d + 1]_q$ linear code. Furthermore, if
	there is an $[n, k, d + 1]_q$ linear code, but no  $[n, k, d + 2]_q$ linear code, we call this $[n, k, d]_q$ linear code \textit{almost (distance)
		optimal}\cite{Xie2025}. The following key theorem clarifies the relationship between the hull dimension of a linear code and its generator matrix.
	
	\begin{theorem}\label{rahu}
		(\cite{Li2019}) Let $\C$ be an $[n,k]$ linear code over $\F_q$ with a generator matrix $G$. Then the dimension $\ell$ of the Euclidean hull of $\C$ is given by
		$$
		\ell=k-\text{rank}(GG^T),
		$$
		where $G^T$ is the transpose of $G$. Similarly, when $q$ is an even power of a prime, the dimension $\ell_h$ of the Hermitian hull of $\C$ is given by
		$$
		\ell_h=k-\text{rank}(GG^\dagger),
		$$
		where $G^\dagger$ is the matrix obtained by taking the $\sqrt{q}$-th power of each entry of $G^T$.
	\end{theorem}
	
	\section{Shortest $t$-dimensional hull embeddings of linear codes when $0\leq t < \ell$}
	Since the Euclidean and Hermitian cases share similar methods in some parts, we adopt the following unified notations.
	\begin{notation}\label{star}
		As in Theorem \ref{rahu}, let $\C$ be an $[n,k]$ linear code over $\F_q$ with $\dim(\text{Hull}(\C))=\ell$. For a matrix $G$ over $\F_q$, we define $G^*$ as follows: $G^*=G^T$ for the Euclidean case and $G^*=G^\dagger$ for the Hermitian case.
	\end{notation}
	
	\begin{definition}
		Let $\C$ be an $[n,k]$ code over $\F_q$. For any $n' \geq n$, an $[n',k]$ code $\widetilde{\C}$ over $\F_q$ is called a \textit{$t$-dimensional Euclidean (resp. Hermitian) hull embedding} of $\C$ if:
		\begin{itemize}
			\item [i)] $\dim (\text{Hull}(\widetilde{\C}))=t$;
			\item [ii)] $\C$ can be obtained by puncturing $\widetilde{\C}$ on some coordinate subset $S$.
		\end{itemize}
	\end{definition}
	We first prove the existence of $t$-dimensional hull embeddings for linear codes. The next proposition holds for both Euclidean and Hermitian hulls.
	\begin{proposition}
		For $0\leq t \leq k$, a $t$-dimensional hull embedding exists for any $[n,k] $ code $\C$ over $\F_q$.
	\end{proposition}
	\begin{proof}
		Let $G$ be a generator matrix of $\C$ and let $p$ be the characteristic of $\F_q$. Consider the matrix
		$$\widetilde{G} = [\underbrace{G,G,\dots,G}_{p},D], $$
		where $D$ is a $k \times (k-t)$ matrix given by
		$$D=\begin{bmatrix} I_{k-t}\\ O\end{bmatrix}.$$
		Then
		$$
		\text{rank}(\widetilde{G}\widetilde{G}^*)=\text{rank}(pGG^*+DD^*)=k-t.
		$$
		According to Theorem \ref{rahu}, the $[np+k-t,k]$ linear code generated by $\widetilde{G}$ is a $t$-dimensional hull embedding of $\C$.
	\end{proof}
	Let $\C$ be an $[n,k]$ linear code over $\F_q$ and let $\widetilde{\C}$ be a $t$-dimensional hull embedding of $\C$. We call $\widetilde{\C}$ a \textit{shortest $t$-dimensional hull embedding} of $\C$ if its length is  minimal among all $t$-dimensional hull embeddings of $\C$. With existence established, a natural question is to find the length of a shortest $t$-dimensional hull embeddings. In the remainder of this section, we first deal with the case $t$  less than the hull dimension of the original linear code.
	\begin{theorem}\label{nlt}
		Let $\C$ be an $[n,k]$ code over $\F_q$ with $\ell =\dim(\text{Hull}(\C))$ and $0 \leq t < \ell$. Then the length of the shortest $t$-dimensional hull embeddings of $\C$ is $n+\ell-t$.
	\end{theorem}
	\begin{proof}
		Let $G$ be a generator matrix of $\C$. Suppose that the length of a shortest $t$-dimensional hull embedding $\widetilde{\C}$ of $\C$ is $n+s$. Then, there exists a $k \times s$ matrix $D$ such that $\widetilde{G}=[G,D]$ generates $\widetilde{\C}$. By Theorem \ref{rahu}, we have
		\begin{equation}
			\begin{aligned}
				k-t=\text{rank}(\widetilde{G}\widetilde{G}^*)&=\text{rank}(GG^*+DD^*)\\
				&\leq \text{rank}(GG^*)+\text{rank}(DD^*)\\
				&\leq \text{rank}(GG^*)+\text{rank}(D)\\
				&\leq k-\ell+s.
			\end{aligned}\label{lmt}
		\end{equation}
		Hence $s \geq \ell-t$. On the other hand, let $G'$ be a generator matrix of $\C$ of the form
		$$G'=\begin{bmatrix} G(\text{Hull}(\C))\\A\end{bmatrix},$$
		where $ G(\text{Hull}(\C))$ is a generator matrix of $\text{Hull}(\C)$. Take a $k \times (\ell-t)$ matrix
		$$D=\begin{bmatrix} I_{\ell-t}\\ O\end{bmatrix}.$$
		We have
		$$[G',D][G',D]^*=\begin{bmatrix} \widetilde{D} & \\  & AA^*\end{bmatrix}$$
		with
		$$ \widetilde{D}=\begin{bmatrix} I_{\ell-t}&\\ & O \end{bmatrix}.$$
		Clearly, $\text{rank}([G',D][G',D]^*)=\text{rank}(AA^*)+\text{rank}(\widetilde{D})=k-t$, and it follows that $[G',D]$ generates a $t$-dimensional hull embedding of $\C$ with length $n+\ell-t$.
	\end{proof}
	From Theorem \ref{nlt}, we also obtain the following conclusion: the length of a shortest $t$-dimensional hull embedding of $\C$ is equal to the length of a shortest $t$-dimensional hull embedding of $ \text{Hull}(\C)$, which is independent of the dimension $k$.
	\begin{corollary}\label{coror}
		Let $\C$ be an $[n,k]$ code over $\F_q$ with $\ell =\dim(\text{Hull}(\C))$, $0\leq t < \ell$, and let $\widetilde{\C}$ be a shortest $t$-dimensional hull embedding of $\C$ with generator matrix
		$$
		\widetilde{G}=\begin{bmatrix} G(\text{Hull}(\C))&D_1 \\A &D_2\end{bmatrix}.
		$$
		Then $\text{rank}(D_1)=\ell -t$.
	\end{corollary}
	\begin{proof}
		Suppose that $r = \text{rank}(D_1) < \ell-t$. Then, the homogeneous linear equation system $XD_1=\textbf{0}$ has $\ell-r$ linearly independent solutions, i.e., there exists an $(\ell -r )\times\ell$ row full-rank matrix $U$ such that $UD_1=\textbf{0}$. Since $ G(\text{Hull}(\C))$ is also row full-rank, we can verify
		$$\text{rank}(U\cdot[G(\text{Hull}(\C)),D_1])=\ell-r.$$
		Furthermore, each row of matrix $U\cdot[G(\text{Hull}(\C)),D_1]$ is a codeword of $\widetilde{\C}$, and all belong to the hull of $\widetilde{\C}$, which is a contradiction. Thus, $\text{rank}(D_1)=\ell-t$.
	\end{proof}
	
	\section{Shortest $t$-dimensional hull embeddings of linear codes when $\ell < t \leq k$}
	In this section, we determine the shortest $t$-dimensional hull embeddings when $t$ exceeds the hull dimension of the original code. We show that the Euclidean and Hermitian inner product structures lead to different conclusions. Recalling Notation \ref{star}, we first present the following concepts.
	\begin{notation}
		For two $n\times n$ matrices $A$ and $B$ over $\F_q$, we say they are congruent if there is an $n \times n$ invertible matrix $D$ such that $DAD^*=B$, and we denote this by $A \simeq B$.
	\end{notation}
	Let $G=[g_{ij}]$ be an $n\times n$ matrix over $F$. We say $G$ is \textit{alternating} if $g_{ii}=0$ and $g_{ij}=-g_{ji}$ for all $1\leq i,j \leq n$. We say $G$ is \textit{symmetric} if $G=G^T$ and \textit{Hermitian} if $G=G^\dagger$. Since matrix congruence is an equivalence relation, we present below the canonical forms of some matrices over finite fields under congruence transformations.
	\begin{proposition}\label{acon}
		(\cite[second part]{Dickson1899}) Let $A$ be an $n\times n$ matrix over $\F_q$ of rank $r$.
		\begin{itemize}
			\item [i)] If $A$ is a Hermitian matrix, then
			\begin{equation*}
				\begin{aligned}
					A\simeq \begin{bmatrix} I_r & \\  & O \end{bmatrix}.
				\end{aligned}
			\end{equation*}
			\item [ii)] If $q$ is an odd prime power and $A$ is a symmetric matrix, then
			\begin{equation*}
				\begin{aligned}
					A\simeq \begin{bmatrix} I_r &  \\  & O \end{bmatrix} \text{ or }\begin{bmatrix} I_{r-1} &  &\\ &z& \\ & & O \end{bmatrix},
				\end{aligned}
			\end{equation*}
			where $z \in \F_q^* \setminus {\F_q^*}^2$. Furthermore, the above two matrices are not congruent.
			\item [iii)] If $A$ is an alternating matrix, then $r$ is even and
			$$A \simeq \text{diag}(\underbrace{J,\dots,J}_{\frac{r}{2}},0,\dots,0), $$\\[-1em]
			where $J=\begin{bmatrix} 0 & 1 \\ -1 & 0 \end{bmatrix}$.
		\end{itemize}
		
	\end{proposition}
	
	\begin{lemma}\label{mgtml}
		Let $\C$ be an $[n,k]$ code over $\F_q$ with $\ell =\dim(\text{Hull}(\C))$ and $\ell < t \leq k$. If the length of the shortest $t$-dimensional hull embeddings of $\C$ is $n+s$, then $s \geq t-\ell$.
	\end{lemma}
	\begin{proof}
		Let $G$ be a generator matrix of $\C$ and let $\widetilde{\C}$ be a shortest $t$-dimensional hull embedding of $\C$. There exists a $k \times s$ matrix $D$ such that $\widetilde{G}=[G,D]$ generates $\widetilde{\C}$. Since
		\begin{equation}
			\begin{aligned}
				k-t=\text{rank}(\widetilde{G}\widetilde{G}^*)&=\text{rank}(GG^*+DD^*)\\
				&\geq \text{rank}(GG^*)-\text{rank}(DD^*)\\
				&\geq \text{rank}(GG^*)-\text{rank}(D)\\
				&\geq k-\ell-s,
			\end{aligned}\label{rageq}
		\end{equation}
		rearranging the inequality gives $s \geq t-\ell$.
	\end{proof}
	The conclusion of Lemma \ref{mgtml} is valid for both Euclidean and Hermitian cases. To find the exact minimal, a case-by-case analysis is required. We first deal with the Hermitian case.
	\begin{theorem}\label{herl}
		Let $\C$ be an $[n,k]$ code over $\F_{q^2}$ with $\ell =\dim(\text{Hull}(\C))$ and $\ell < t \leq k$. Then the length of the shortest $t$-dimensional \textbf{Hermitian} hull embeddings of $\C$ is $n+t-\ell$.
	\end{theorem}
	\begin{proof}
		Let $G$ be a generator matrix of $\C$. Then $GG^\dagger$ is a Hermitian matrix and $\text{rank}(GG^{\dagger})=k-\ell$ by Theorem \ref{rahu}. According to Proposition \ref{acon} i), there exists a $k\times k$ invertible matrix $H$ such that
		$$GG^\dagger=H\begin{bmatrix} I_{k-\ell} & \\  & O \end{bmatrix}H^\dagger.$$
		Let us take a $k\times (t-\ell)$ matrix
		$$
		D=H\begin{bmatrix} aI_{t-\ell} \\  O\end{bmatrix}
		$$
		with $a^{q+1}=-1$. Clearly, $\text{rank}([G,D][G,D]^\dagger)=k-t$ and it follows that $[G,D]$ generates a $t$-dimensional Hermitian hull embedding of $\C$.
	\end{proof}
	We next study the Euclidean hull embeddings.
	\begin{lemma}\label{qec}
		Let $q=2^m$ and let $A$ be an $n\times n$ symmetric matrix over $\F_q$ of rank $r$.
		\begin{itemize}
			\item [i)]If $A$ is alternating, then $r$ is even and
			$$A \simeq \text{diag}(\underbrace{J,\dots,J}_{\frac{r}{2}},0,\dots,0), $$\\[-1em]
			where $J=\begin{bmatrix} 0 & 1 \\ 1 & 0 \end{bmatrix}$.
			\item [ii)]If $A$ is non-alternating, then
			\begin{equation*}
				\begin{aligned}
					A\simeq \begin{bmatrix} I_r  & \\ & O
					\end{bmatrix} .
				\end{aligned}
			\end{equation*}
		\end{itemize}
		
	\end{lemma}
	\begin{proof}
		i) When $G$ is alternating, the conclusion is clear by Proposition \ref{acon} iii).
		
		ii) Now we suppose $G$ is non-alternating, which means there exists a non-zero element on the diagonal of $G$. Up to congruence, we have
		\begin{equation}
			\begin{aligned}
				G\simeq \begin{bmatrix} 1 & \alpha \\ \alpha^T & G' \end{bmatrix}\simeq \begin{bmatrix} 1 &  \\  & G'' \end{bmatrix},
			\end{aligned}\label{repe}
		\end{equation}
		where $G''=G'-\alpha^T\alpha$ is symmetric.\\
		\textbf{Case 1.} If $G''$ is alternating, then
		$$G'' \simeq \text{diag}(J,\dots,J,0,\dots,0), $$
		and
		$$
		G \simeq \text{diag}(1,J,\dots,J,0,\dots,0).
		$$
		\textbf{Case 2.} If $G''$ is non-alternating, repeat the steps in \eqref{repe}. In any case, we always have
		\begin{equation}
			G \simeq \text{diag}(I_j,J,\dots,J,0,\dots,0)\label{ijj}
		\end{equation}
		$$ $$
		and $j \geq 1$. Furthermore, denote by
		$$
		P=\begin{bmatrix} 1 & 1 & 1 \\ 1 & 1 & 0 \\ 1 & 0 & 1 \end{bmatrix}.
		$$
		Then $P$ is invertible and one can verify
		$$
		P \begin{bmatrix} 1& \\   & J  \end{bmatrix} P^T= I_3.
		$$
		Using this relationship, we can gradually reduce the $J$-blocks in \eqref{ijj} to the identity matrix. Thus, the desired result follows.
	\end{proof}
	Here we set $\F_q^*=\{a\in \F_q : a \neq 0\}$, ${\F_q^*}^2=\{a^2 : a \in \F_q^* \}$ and $|\F_q^*:{\F_q^*}^2|=2$. Recalling Proposition \ref{acon} ii) and Lemma \ref{qec}, we present the following classification of linear codes.
	\begin{definition}\label{types}
		Let $\C$ be an $[n,k]$ linear code over $\F_q$ with generator matrix $G$.
		\begin{itemize}
			\item [i)] Let $q$ be an odd prime power. Then $\C$ is said to be of type  \textbf{($\text{E}_{\text{o,s}}$)} if
			$$-GG^T \simeq \begin{bmatrix} I_r &  \\  & O \end{bmatrix};$$
			of type  \textbf{($\text{E}_{\text{o,ns}}$)} if
			$$-GG^T \simeq \begin{bmatrix} I_{r-1} &  &\\ &z& \\ & & O \end{bmatrix},$$
			where $z \in \F_q^* \setminus {\F_q^*}^2$.
			\item [ii)] Let $q=2^m$. Then $\C$ is said to be of type \textbf{($\text{E}_{\text{e,a}}$)} if $GG^T$ is alternating; of type \textbf{($\text{E}_{\text{e,na}}$)} if $GG^T$ is non-alternating.
		\end{itemize}
		
	\end{definition}
	\begin{lemma}\label{equ}
		Let $q$ be an odd prime power, and let $G$ be an invertible symmetric $n\times n$ matrix over $\F_q$. Then there exists an $n\times n$ invertible matrix $B$ such that $G=BB^T$ if and only if $\det(G) \in {\F_q^*}^2 $.
	\end{lemma}
	\begin{proof}
		The necessity is obvious. If $\det(G) \in {\F_q^*}^2 $, the sufficiency can be obtained by Proposition \ref{acon} ii).
	\end{proof}
	\begin{lemma}\label{appt}
		Let $q=2^m$. If there exists a $k \times s$ matrix $P$ over $\F_q$ such that $PP^T$ is alternating, then we have $\text{rank}(PP^T)\leq s-1$.
	\end{lemma}
	\begin{proof}
		Consider the set $V=\{\bu\in \F_q^s: \left \langle \bu,\bu \right \rangle_E=0\}$. Actually
		$$V=\{(u_1,\dots,u_s)\in \F_q^s: \sum_{i=1}^{s}u_i=0\}, $$
		since $\text{char }\F_q=2$. Note that each row vector of the matrix \(P\) belongs to $V$. Then
		$$\text{rank}(PP^T) \leq \text{rank}(P) \leq \dim V= s-1.$$
	\end{proof}
	\begin{lemma}\label{eea}
		Let $\C$ be an $[n,k]$ code over $\F_{q}$ with $\ell =\dim(\text{Hull}(\C))$ and $\ell < t \leq k$. Suppose the length of the shortest $t$-dimensional \textbf{Euclidean} hull embeddings of $\C$ is $n+s$. Then we have $s \geq t-\ell+1$ if $\C$ is of type \textbf{($\text{E}_{\text{e,a}}$)} defined in Definition \ref{types} ii).
	\end{lemma}
	\begin{proof}
		Let $G$ be a generator matrix of $\C$. W.L.O.G., set (note $\text{rank}(GG^T)= k-\ell$ is even)
		$$B=GG^T=\text{diag}(\underbrace{J,\dots,J}_{\frac{k-\ell}{2}},0,\dots,0),$$\\ [-1.3em]
		where $J=\begin{bmatrix} 0 & 1 \\ 1 & 0 \end{bmatrix}$.
		Combining Lemma \ref{mgtml}, we only need to prove that $\C$ has no $t$-dimensional Euclidean hull embedding with length $n+t-\ell$. Otherwise, there exists a $k \times (t-\ell)$ matrix $D$ such that
		$$\text{rank}([G,D][G,D]^T)=\text{rank}(B+DD^T)=k-t.$$
		According to the rank inequality,
		$$
		t-\ell \geq \text{rank}(D) \geq \text{rank}(DD^T) \geq \text{rank}(B)- \text{rank}(B+DD^T)=t-\ell,
		$$
		which implies $\text{rank}(DD^T)=\text{rank}(D)=t-\ell $.
		
		\textbf{case 1.} $DD^T$ is alternating. A contradiction can be derived from Lemma \ref{appt}.
		
		\textbf{case 2.} $DD^T$ is non-alternating. Set the following spaces
		$$
		\begin{aligned}
			U&=\{ \bu\in \F_q^k: (B+DD^T)\bu=0\},\\
			V&=\{ \bu\in \F_q^k: D^T\bu=0\},\\
			W&=\{ \bu\in \F_q^k: B\bu=0\}.
		\end{aligned}
		$$
		Consider the linear map $\varphi: U \to \F_q^{(t-\ell)}$, $\varphi(\bu)=D^T\bu$ for all $\bu \in U $. Since $B$ is alternating,
		$$\text{im } \varphi \subseteq \{ \bu\in \F_q^{(t-\ell)}: \left \langle \bu,\bu\right \rangle_E=0\}. $$
		One can derive $\dim(\text{im } \varphi) \leq t-\ell-1$. Furthermore,
		we can verify
		\begin{equation}
			\ker \varphi=V \cap W.\label{cap}
		\end{equation}
		Then, according to the Rank-Nullity theorem, we have
		\begin{equation*}
			\dim(\text{im } \varphi)+\dim(\ker \varphi)= \dim (U).
		\end{equation*}
		Combining with $\dim (U)=k-\text{rank}(B+DD^T)=t$, we derive
		$$
		\begin{aligned}
			\dim(\ker \varphi) &= \dim (U)-\dim(\text{im } \varphi)\\
			&\geq t -(t-\ell-1)\\
			&=\ell+1.
		\end{aligned}
		$$
		By \eqref{cap}, $\ell+1 \leq \dim(\ker \varphi) \leq \dim (W)=\ell$, which is a contradiction. Thus, the desired conclusion holds.
	\end{proof}
	\begin{lemma}\label{constr}
		For $q=2^m$ and any $r \in \mathbb{N}^+$, there always exists a $2r \times (2r+1)$ matrix $P$ over $\F_q$ such that
		$$
		PP^T=\text{diag}(\underbrace{J,\dots,J}_{r}),
		$$\\ [-1em]
		where $J=\begin{bmatrix}
			0&1\\
			1&0
		\end{bmatrix}.$
	\end{lemma}
	
	\begin{proof}
		We proceed by induction on $r$. The $i\times j$ all-ones matrix is denoted by $\textbf{1}_{i,j}$. When $r=1$, take the $2\times 3$ matrix
		$$
		P_1=\begin{bmatrix}
			1&1&0\\
			1&0&1
		\end{bmatrix}.
		$$
		Then $P_1P_1^T=J$ and $P_1\textbf{1}_{3,2}=O$. Suppose there exists a $2r \times(2r+1)$ matrix $P_r$ such that
		$$
		P_rP_r^T=\text{diag}(\underbrace{J,\dots,J}_{r}) \  \text{ and }\  P_r\textbf{1}_{2r+1,2}=O.
		$$
		Now taking the $(2r+2) \times(2r+3)$ matrix
		$$
		P_{r+1}=\begin{bmatrix}
			P_r&O\\
			\textbf{1}_{2,2r+1}&J
		\end{bmatrix},
		$$
		we have
		$$
		P_{r+1}P_{r+1}^T=\begin{bmatrix}
			P_rP_r^T&P_r\textbf{1}_{2r+1,2}\\
			P_r^T\textbf{1}_{2,2r+1}&\textbf{1}_{2,2r+1}\textbf{1}_{2r+1,2}+J^2
		\end{bmatrix}=\text{diag}(\underbrace{J,\dots,J}_{r+1}),
		$$
		and
		$$
		P_{r+1}\textbf{1}_{2r+3,2}=\begin{bmatrix}
			P_r&O\\
			\textbf{1}_{2,2r+1}&J
		\end{bmatrix}
		\begin{bmatrix}
			\textbf{1}_{2r+1,2}\\
			\textbf{1}_{2,2}
		\end{bmatrix}=O.
		$$
		Thus, the conclusion follows by induction.
	\end{proof}
	\begin{theorem}\label{sel}
		Let $\C$ be an $[n,k]$ code over $\F_{q}$ with $\ell =\dim(\text{Hull}(\C))$ and $\ell < t \leq k$. If the length of the shortest $t$-dimensional \textbf{Euclidean} hull embeddings of $\C$ is $n+s$, then:
		
		i) For the case where $q$ is an odd prime power,
		\begin{equation}
			s =
			\begin{cases}
				t-\ell,&\text{$\C$ is of type  \textbf{($\text{E}_{\text{o,s}}$)}}, \\
				t-\ell,&\text{$\C$ is of type  \textbf{($\text{E}_{\text{o,ns}}$)} and } \ell \leq t\leq k-1, \\
				k-\ell+1,&\text{$\C$ is of type \textbf{($\text{E}_{\text{o,ns}}$)} and }t=k.
			\end{cases}\label{qodd}
		\end{equation}
		
		ii) For the case where $q=2^m$,
		\begin{equation}
			s =
			\begin{cases}
				t-\ell,&\text{$\C$ is of type \textbf{($\text{E}_{\text{e,na}}$)}}, \\
				t-\ell+1,&\text{$\C$ is of type \textbf{($\text{E}_{\text{e,a}}$)}}.
			\end{cases}\label{qeven}
		\end{equation}
	\end{theorem}
	\begin{proof}
		Let $G$ be a generator matrix of $\C$. Our goal is to find the corresponding $k\times s $ matrix $D$ for each type such that
		\begin{equation}
			\text{rank}([G,D][G,D]^T)=k-t.\label{gdkt}
		\end{equation}
		Note that $\text{rank}(GG^T)=k-\ell > k-t$.
		
		(a) For type \textbf{($\text{E}_{\text{e,na}}$)} or \textbf{($\text{E}_{\text{o,s}}$)}, by Lemma \ref{qec} and Definition \ref{types}, there exists an invertible $k\times k$ matrix $B$ such that
		$$
		-GG^T = B\begin{bmatrix} I_{k-\ell}  & \\ & O
		\end{bmatrix}B^T.
		$$
		Therefore, we take the $k \times (t-\ell)$ matrix
		$$D =B\begin{bmatrix} I_{t-\ell}\\ O
		\end{bmatrix} $$
		satisfying \eqref{gdkt}.
		
		(b) For type \textbf{($\text{E}_{\text{o,ns}}$)} and $\ell < t \leq k-1$, according to Proposition \ref{acon} ii), there exists an invertible $k\times k$ matrix $C$ such that
		$$
		-GG^T = C\begin{bmatrix} I_{k-\ell-1} &  &\\ &z& \\ & & O \end{bmatrix}C^T,
		$$
		where $z \in \F_q^* \setminus {\F_q^*}^2$. Note that $t-\ell \leq k-\ell-1$. Then we can take
		$$D =C\begin{bmatrix} I_{t-\ell}\\ O
		\end{bmatrix}$$
		satisfying \eqref{gdkt}.
		
		(c) For type \textbf{($\text{E}_{\text{o,ns}}$)} and $t = k$, let $G'$ be a generator matrix of $\C$ of the form
		$$G'=\begin{bmatrix} G(\text{Hull}(\C))\\A\end{bmatrix},$$
		where $G(\text{Hull}(\C)) $ is a generator matrix of $\text{Hull}(\C)$. Note that $G$ and $G'$ are row equivalent. Therefore
		$$
		\begin{bmatrix} O & \\ & -AA^T  \end{bmatrix}=-G'G'^T \simeq-GG^T \simeq \begin{bmatrix} I_{k-\ell-1} &  &\\ &z& \\ & & O \end{bmatrix},
		$$
		where $z \in \F_q^* \setminus {\F_q^*}^2$ and one can verify
		\begin{equation}
			-AA^T \simeq \begin{bmatrix} I_{k-\ell-1} &  \\ & z \end{bmatrix}.\label{iz}
		\end{equation}
		Assume $s=k-\ell$.  Then there exists a $k\times (k-\ell) $ matrix $D$ such that
		$$
		[G',D][G',D]^T=O.
		$$
		Partition $D$ into blocks:
		$$
		D=\begin{bmatrix} D_1 \\ D_2\end{bmatrix},
		$$
		where $D_2$ is a $(k-\ell) \times (k-\ell) $ matrix. This means
		$$
		D_2D_2^T=-AA^T\simeq \begin{bmatrix} I_{k-\ell-1} &  \\ & z \end{bmatrix}
		$$
		by \eqref{iz}. Then $\det(D_2D_2^T) \in \F_q^* \setminus {\F_q^*}^2$ which contradicts Lemma \ref{equ}. Therefore, $s \geq k-\ell+1$. Consider the set
		$$S_z=\{z-a^2: a \in \F_q\}. $$
		We know $\#S_z=\#\{a^2:a \in \F_q\}=\frac{q+1}{2}$. Then,
		$$S_z \cap \{a^2:a \in \F_q\} \neq \emptyset,$$
		and there exists $z_1,z_2 \in \F_q$ such that $z=z_1^2+z_2^2$. By \eqref{iz}, there exists an invertible $(k-\ell)\times (k-\ell)$ matrix $F$ such that
		$$
		-AA^T=F\begin{bmatrix} I_{k-\ell-1} &  \\ & z \end{bmatrix} F^T.
		$$
		Now we take the $k\times(k-\ell+1)$ matrix
		$$
		D=\begin{bmatrix} O \\ D' \end{bmatrix}
		$$
		with the $(k-\ell)\times (k-\ell+1)$ matrix
		$$D'=F\begin{bmatrix} I_{k-\ell-1}& \textbf{0}&\textbf{0} \\ \textbf{0} & z_1&z_2 \end{bmatrix}.$$
		One can verify
		$$
		[G',D][G',D]^T= \begin{bmatrix}  O& \\ &AA^T+D'D'^T \end{bmatrix}=O,
		$$
		and it follows that $[G',D]$ generates a $k$-dimensional Euclidean hull embedding of $\C$.
		
		For type \textbf{($\text{E}_{\text{e,a}}$)}, by Lemma \ref{eea}, we only need to find a $k \times (t-\ell+1)$ matrix $D$ satisfying \eqref{gdkt}.
		
		(d) For type \textbf{($\text{E}_{\text{e,a}}$)} and $t-\ell$ is odd, W.L.O.G., set (note $\text{rank}(GG^T)= k-\ell$ is even)
		$$GG^T=  \text{diag}(\underbrace{J,\dots,J}_{\frac{k-\ell}{2}},0,\dots,0), $$
		where $J=\begin{bmatrix} 0 & 1 \\ 1 & 0 \end{bmatrix}$. It is clear from the proof of Lemma \ref{qec} that the $(t-\ell)\times (t-\ell) $ matrix
		$$M:=  \text{diag}(\underbrace{J,\dots ,J}_{\frac{t-\ell-1}{2}},1)\simeq I_{t-\ell}. $$
		Then there exists an invertible $(t-\ell)\times (t-\ell)$ matrix $D'$ such that
		$$
		M=D'D'^T.
		$$
		Now we take $k\times(t-\ell+1)$ matrix
		\begin{equation*}
			\begin{aligned}
				D= & \begin{bmatrix}
					D' &\textbf{0} \\
					\textbf{0}& 1 \\
					O&\textbf{0}
				\end{bmatrix}.
			\end{aligned}
		\end{equation*}
		It derives
		\begin{equation*}
			\begin{aligned}
				[G,D][G,D]^T&=\text{diag}(\underbrace{J,\dots,J}_{\frac{k-\ell}{2}},0,\dots,0)+\text{diag}(\underbrace{J,\dots,J}_{\frac{t-\ell-1}{2}},1,1,0,\dots,0)\\
				&=\text{diag}(\underbrace{0,\dots,0}_{t-\ell-1},J+I_2,\underbrace{J,\dots,J}_{\frac{k-t-1}{2}},0,\dots,0).
			\end{aligned}
		\end{equation*}
		Obviously, $\text{rank}([G,D][G,D]^T)=k-t$ which satisfies \eqref{gdkt}. Then $[G,D]$ generates a shortest $t$-dimensional hull embedding of $\C$.
		
		(e) For type \textbf{($\text{E}_{\text{e,a}}$)} and $t-\ell$ is even, a matrix $D$ satisfying \eqref{gdkt} can be constructed recursively according to Lemma \ref{constr}.
	\end{proof}
	\begin{remark}
		Notably, compared with the results of An et al. \cite{An2025} for binary linear codes, our conclusion in Theorem \ref{sel} provides a criterion for precisely determining the length of the shortest self-orthogonal embeddings, namely that the length depends only on whether the binary linear code is alternating.
	\end{remark}
	The following result corresponds to Corollary \ref{coror}.
	\begin{corollary}
		Let $\C$ be an $[n,k]$ code over $\F_q$ with $\ell =\dim(\text{Hull}(\C))$, $\ell < t \leq k$, and let $\widetilde{\C}$ be a shortest $t$-dimensional hull embedding of $\C$ with a generator matrix
		$$
		\widetilde{G}=\begin{bmatrix} G(\text{Hull}(\C))&D_1 \\A &D_2\end{bmatrix}.
		$$
		Then
		\begin{equation}
			\text{rank}(D_2) =
			\begin{cases}
				t-\ell, &\text{in the Hermitian case,}\\
				t-\ell,&\text{$\C$ is of type \textbf{($\text{E}_{\text{o,s}}$)}, \textbf{($\text{E}_{\text{o,ns}}$)}, or \textbf{($\text{E}_{\text{e,na}}$)}}, \\
				t-\ell \text{ or }t-\ell+1,&\text{$\C$ is of type \textbf{($\text{E}_{\text{e,a}}$)}}.
			\end{cases}\label{rd2}
		\end{equation}
		
	\end{corollary}
	\begin{proof}
		Recall Notation \ref{star}. Consider the following spaces
		$$
		\begin{aligned}
			V_1&=\{\bu=(\bu_1,\textbf{0})\in \F_q^k:\bu_1\in \F_q^{\ell},\ \bu\widetilde{G}\widetilde{G}^*=\textbf{0}\},\\
			V_2&=\{\bu=(\textbf{0},\bu_2) \in \F_q^k:\bu_2\in \F_q^{(k-\ell)},\ \bu\widetilde{G}\widetilde{G}^*=\textbf{0}\},\\
			\widetilde{V}&=\{\bu \in \F_q^k:\bu\widetilde{G}\widetilde{G}^*=\textbf{0}\}.
		\end{aligned}
		$$
		Clearly, $V_1 \cap V_2=\{\textbf{0}\}$ and $\widetilde{V}=V_1\oplus V_2 $. Hence
		\begin{equation}
			\dim V_2 =\dim(V_1\oplus V_2)-\dim V_1 \geq t-\ell.\label{vpg}
		\end{equation}
		Define a linear map $\phi: V_2 \to \F_q^m$, $\phi(\textbf{0},\bu_2)=\bu_2D_2$ for all $(\textbf{0},\bu_2) \in V_2$. We next prove that $\phi$ is surjective. Suppose $\phi(\textbf{0},\bu_2)=\bu_2D_2=\textbf{0}$ for $(\textbf{0},\bu_2) \in V_2 $; we have
		$$
		\textbf{0}=\bu_2(AA^*+D_2D_2^*)=\bu_2AA^*,
		$$
		since $(\textbf{0},\bu_2)\widetilde{G}\widetilde{G}^*=\textbf{0}$. Note that $AA^*$ is invertible. Hence $\bu_2=\textbf{0}$. Then
		\begin{equation}
			\dim V_2=\dim (\text{im} \phi) \leq \text{rank}(D_2).\label{vpl}
		\end{equation}
		Combining \eqref{vpg} and \ref{vpl}, we obtain
		$$t-\ell \leq \text{rank}(D_2)\leq  s \leq t-\ell+1.$$
		
		i) According to Theorems \ref{herl} and \ref{sel}, we get $\text{rank}(D_2)=t-\ell$ for the Hermitian case, type \textbf{($\text{E}_{\text{o,s}}$)} and \textbf{($\text{E}_{\text{e,na}}$)}.
		
		ii) When $\C$ is of type \textbf{($\text{E}_{\text{o,ns}}$)} and $\ell <t \leq k-1$, we also have $\text{rank}(D_2)=t-\ell$ by \eqref{qodd} in Theorem \ref{sel}. When $\C$ is of type \textbf{($\text{E}_{\text{o,ns}}$)} and $t = k$, note that $D_2$ has only $t-\ell$ rows. Hence $\text{rank}(D_2)=k-\ell$, i.e., $D_2$ is invertible.
		
		iii) For type \textbf{($\text{E}_{\text{e,a}}$)}, referring to the construction in Theorem \ref{sel}, we see that both cases in \eqref{rd2} can occur.
		
	\end{proof}
	
	\section{Construction algorithm and examples}
	 The classification and constructions of optimal codes have been explored in \cite{An2025,Kim2021,Li2008}. In fact, many optimal linear codes in the \textbf{MAGMA} Best Known Linear Code (BKLC) database\cite{Magma1994} are self-orthogonal. We now present construction algorithms for the shortest $t$-dimensional hull embeddings in both the Euclidean and Hermitian settings, along with illustrative examples.
	\subsection{Hermitian case}
	In this subsection, we present an algorithm for constructing a shortest $t$-dimensional Hermitian hull embedding of a $q^2$-ary linear code, followed by concrete examples.
	
	\begin{algorithm}
		\caption{Construction of a shortest $t$-dimensional Hermitian hull embedding of a $q^2$-ary linear code}
		\begin{algorithmic}[1]
			\State	Let $\C$ be an $[n,k]$ linear code over $\F_{q^2}$ with generator matrix $G$ and $0 \leq t \leq k$
			\State \textbf{Goal:} Find a $k \times (n+s)$ matrix $\widetilde{G}$ over $\F_{q^2}$ which generates a shortest $t$-dimensional Hermitian hull embedding of $\C$
			\Procedure{HermitianHullEmbedding}{$G$, $t$}
			\State $k \gets \text{the number of rows in matrix } G$
			\State $\ell \gets k-\text{rank}(GG^\dagger)$
			\State Precompute a $k\times k$ invertible matrix $P$ s.t. $GG^\dagger=P \text{diag}(I_{k-\ell},0,\dots,0) P^\dagger$
			\State $s \gets |t-\ell|$
			\If{$t <\ell$}
			\State  $D \gets P\begin{bmatrix} O_{(k-s) \times s} \\ I_s \end{bmatrix}$
			\Else
			\State Select $a\in \{x \in \F_{q^2}: x^{q+1}=-1\}$
			\State  $D \gets P \begin{bmatrix} a I_s \\ O_{(k-s) \times s} \end{bmatrix}$
			\EndIf
			\State $\widetilde{G} \gets [G\mid D]$
			\State \Return $\widetilde{G}$
			\EndProcedure
		\end{algorithmic}
	\end{algorithm}
	
	We denote the $q$-ary $[\frac{q^r-1}{q-1},\frac{q^r-1}{q-1}-r,3]$ Hamming code by $\mathcal{H}_{q,r}$. We give the following two examples to demonstrate how to apply Algorithm 1.
	\begin{example}\label{634}
		Let $\zeta$ be a primitive element of $\F_4 $ ($\zeta^2=\zeta+1$) and let $\mathcal{H}_{4,2}$ be the $[5,3,3]_4$ Hamming code. There are two matrices
		$$
		G=\begin{bmatrix} 1&0&0&1&\zeta^2 \\
			0&1&0&\zeta^2&\zeta^2\\ 0&0&1&\zeta^2&1\end{bmatrix}\text{ and }
		P=\begin{bmatrix} 1&0&0 \\
			\zeta&1&0\\ 1&0&1\end{bmatrix}=[P_1,P_2,P_3],
		$$
		where $G$ is a generator matrix of $\mathcal{H}_{4,2}$ and $P$ is an invertible matrix such that $$GG^\dagger=P\text{diag}(1,0,0)P^\dagger .$$
		We have  $\dim(\text{Hull}(\mathcal{H}_{4,2}))=3-\text{rank}(GG^\dagger)=2$. Table 1 shows how the parameters vary after applying Algorithm 1 to obtain its shortest $t$-dimensional Hermitian hull embeddings (HSO stands for Hermitian self-orthogonal).
		\begin{table}[H]
			\centering
			\vspace{-0.2em}
			\renewcommand{\arraystretch}{1.1}
			\begin{tabular}{|c|c|c|}
				\hline
				$t$ & Generator matrix & $[n,k,d]$    \\
				\hline
				$0$ (LCD)&$[G,P_2,P_3]$&$[7,3,3]$ \\
				\hline
				$1$	 &$[G,P_3]$ & $[6,3,3]$    \\
				\hline
				$3$ (HSO)&$[G,P_1]$&$\textbf{[6,3,4]}$ \\
				\hline
			\end{tabular}
			\caption{Shortest $t$-dimensional Hermitian hull embeddings of $\mathcal{H}_{4,2}$}
		\end{table}
		Comparing with the BKLC database, the $[6,3,4]_4$ linear code generated by $[G,P_1]$ is optimal and MDS.
	\end{example}
	\begin{example}
		Let $w$ be a primitive element of $\F_9 $ ($w^2=w+1$). Let $\C_{8,4}$ be an $[8,4,5]_9$ MDS code. There are two matrices
		$$
		G=\begin{bmatrix} 1&0&0&0&2&1&w&w^2 \\
			0&1&0&0&1&w^6&2&w^7	\\
			0&0&1&0&w^5&w^7&w^7&w^7	\\
			0&0&0&1&w^5&1&2&w^6
		\end{bmatrix},
		$$
		and
		$$
		P=\begin{bmatrix} w^2&w^3&0&0 \\
			1&w^7&0&0	\\
			w^4&w^6&1&0	\\
			w^2&w^2&w^4&1
		\end{bmatrix}=[P_1,P_2,P_3,P_4],
		$$
		where $G$ is a generator matrix of $\C_{8,4}$ and $P$ is an invertible matrix such that $$GG^\dagger=P\text{diag}(1,1,1,0)P^\dagger.  $$
		We have $\dim(\text{Hull}(\C_{8,4}))=4-\text{rank}(GG^\dagger)=1$. Using Algorithm 1 with $w^4=-1$, we obtain Table 2.
		\begin{table}[H]
			\centering
			\vspace{-0.2em}
			\renewcommand{\arraystretch}{1.1}
			\begin{tabular}{|c|c|c|}
				\hline
				$t$&Generator matrix &$[n,k,d]$ \\
				\hline
				$0$ (LCD)&$[G,P_4]$ & $[9,4,5]$\\
				\hline
				$2$&$[G,wP_1]$ &$[9,4,5]$ \\
				\hline
				$3$&$[G,wP_1,wP_2]$ &$[10,4,5]$ \\
				\hline
				$4$ (HSO)&$[G,wP_1,wP_2,wP_3]$ &$\textbf{[11,4,6]}$ \\
				\hline										
			\end{tabular}
			\caption{Shortest $t$-dimensional Hermitian hull embeddings of $\C_{8,4}$}
		\end{table}
		Comparing with the BKLC database, the $[11,4,6]_9$ linear code generated by $[G,wP_1,wP_2,wP_3]$ is almost optimal.
	\end{example}
	\subsection{Euclidean case}
	In this subsection, we will show two algorithms and some examples on constructing the shortest $t$-dimensional Euclidean hull embeddings of $q$-ary linear codes. We first present the construction algorithm and examples for $q$ being an odd prime power.
	\begin{algorithm}
		\caption{Construction of a shortest $t$-dimensional Euclidean hull embedding of $q$-ary linear code when $q$ is an odd prime power }
		\begin{algorithmic}[1]
			\State Let $q$ be an odd prime power and let $\C$ be an $[n,k]$ linear code over $\F_{q}$ with a generator matrix $G$ and $0 \leq t \leq k$
			\State \textbf{Goal:} Find a $k \times (n+s)$ matrix $\widetilde{G}$ over $\F_{q}$ that generates a shortest $t$-dimensional Euclidean hull embedding of $\C$
			\Procedure{OddEuclideanHullEmbedding}{$G$, $t$}
			\State $k \gets \text{the number of rows in matrix } G$
			\State $\ell \gets k-\text{rank}(GG^T)$
			\State Precompute a $k \times k$ invertible matrix $P$ s.t. $-GG^T=P\text{diag}(I_{k-\ell-1},z,0,\dots,0)P^T$
			\State  $s \gets  |t-\ell| $
			\State $D \gets  O_{s\times (k-s)}$
			\If{$t < \ell$}
			\State  $D \gets [D,I_s]^T  $
			\ElsIf{$z \in {\F_q^*}^2$}
			\State $P \gets P\text{diag}(I_{k-\ell-1},z^{\frac{1}{2}},0,\dots,0)$
			\State $D \gets  P[I_s,D]^T $
			\Else
			\If{$\ell <t \leq k-1$}
			\State $D \gets  P[I_s,D]^T$
			\Else
			\State Compute $z_1,z_2 \in \F_q$ s.t. $z=z_1^2+z_2^2$
			\State $D'\gets \begin{bmatrix} I_{k-\ell-1}& \textbf{0}&\textbf{0} \\ \textbf{0} & z_1&z_2 \end{bmatrix}$
			\State  $D \gets  P\begin{bmatrix} D'  \\O_{\ell\times (s+1)} \end{bmatrix}$
			\EndIf
			\EndIf
			\State $\widetilde{G} \gets [G, D]$
			\State \Return $\widetilde{G}$
			\EndProcedure
		\end{algorithmic}
	\end{algorithm}
	\clearpage
	\begin{example}\label{2089}
		Let $\C_{18,8}$ be an $[18,8,7]$ ternary optimal code (from BKLC). There are two matrices
		\[G= \left[
		\setlength{\arraycolsep}{3pt}
		\renewcommand{\arraystretch}{0.7}
		\begin{array}{cccccccccccccccccc}
			1 & 0 & 0 & 0 & 0 & 0 & 0 & 0 & 0 & 1 & 1 & 0 & 1 & 2 & 2 & 0 & 1 & 0 \\
			0 & 1 & 0 & 0 & 0 & 0 & 0 & 0 & 0 & 2 & 0 & 1 & 2 & 2 & 0 & 2 & 2 & 1 \\
			0 & 0 & 1 & 0 & 0 & 0 & 0 & 0 & 1 & 1 & 1 & 1 & 2 & 2 & 2 & 0 & 2 & 2 \\
			0 & 0 & 0 & 1 & 0 & 0 & 0 & 0 & 2 & 2 & 1 & 0 & 2 & 0 & 0 & 2 & 2 & 2 \\
			0 & 0 & 0 & 0 & 1 & 0 & 0 & 0 & 2 & 2 & 1 & 0 & 0 & 1 & 2 & 0 & 0 & 2 \\
			0 & 0 & 0 & 0 & 0 & 1 & 0 & 0 & 2 & 0 & 2 & 0 & 1 & 1 & 2 & 2 & 2 & 0 \\
			0 & 0 & 0 & 0 & 0 & 0 & 1 & 0 & 0 & 2 & 0 & 2 & 0 & 1 & 1 & 2 & 2 & 2 \\
			0 & 0 & 0 & 0 & 0 & 0 & 0 & 1 & 2 & 0 & 1 & 2 & 2 & 2 & 0 & 1 & 0 & 2
		\end{array}\right],
		\]
		and
		$$
		P=\left[
		\setlength{\arraycolsep}{3pt}
		\renewcommand{\arraystretch}{0.7}
		\begin{array}{ccccccccccc}
			2 & 2 & 0 & 0 & 0 & 0 & 0 & 0 \\
			0 & 1 & 0 & 0 & 0 & 0 & 0 & 0 \\
			1 & 1 & 0 & 0 & 0 & 0 & 0 & 1 \\
			1 & 0 & 1 & 0 & 0 & 0 & 0 & 0 \\
			2 & 1 & 0 & 1 & 0 & 0 & 0 & 0 \\
			2 & 0 & 0 & 0 & 1 & 0 & 0 & 0 \\
			2 & 0 & 0 & 0 & 0 & 1 & 0 & 0 \\
			0 & 1 & 0 & 0 & 0 & 0 & 1 & 0
		\end{array}\right]=[P_1,P_2,\dots,P_8],
		$$
		where $G$ is a generator matrix of $\C_{18,8}$ and $P$ is an invertible matrix such that $$-GG^T=P\text{diag}(1,1,0,0,0,0,0,0)P^T .$$ Therefore,  $\dim(\text{Hull}(\C_{18,8}))=8-\text{rank}(GG^T)=6$ and $\C_{18,8}$ is of type \textbf{($\text{E}_{\text{o,s}}$)}. Using Algorithm 2, we obtain Table 3 (ESO stands for Euclidean self-orthogonal).
		\begin{table}[H]
			\centering
			\vspace{-0.2em}
			\renewcommand{\arraystretch}{1.1}
			\begin{tabular}{|c|c|c|c|c|c|}
				\hline
				$t$ & Generator matrix & $[n,k,d]$ &  $t$ & Generator matrix & $[n,k,d]$ \\
				\hline
				$0$ (LCD)&$[G,P_3,P_4,P_5,P_6,P_7,P_8]$ &$[24,8,7]$ & $4$&$[G,P_7,P_8]$ &$[20,8,7]$\\
				\hline
				$1$& $[G,P_4,P_5,P_6,P_7,P_8]$ & $[23,8,7]$ &  $5$ &$[G,P_8]$ &$[19,8,7]$ \\
				\hline
				$2$ &$[G,P_5,P_6,P_7,P_8]$ &$[22,8,7]$ & $7$& $[G,P_1]$&$\textbf{[19,8,8]}$\\
				\hline
				$3$ &$[G,P_6,P_7,P_8]$ &$[21,8,7]$ & $8$ (ESO)&$[G,P_1,P_2]$ &$\textbf{[20,8,9]}$ \\
				\hline
				
			\end{tabular}
			\caption{Shortest $t$-dimensional Euclidean hull embeddings of $\C_{18,8}$}
		\end{table}
		
		By comparison with the BKLC database, the $[19,8,8]_3$ linear code generated by $[G,P_1]$ is optimal and is not equivalent to BKLC(GF(3),19,8); the $[20,8,9]_3$ linear code generated by $[G,P_1,P_2]$ is optimal and is not equivalent to BKLC(GF(3),20,8).
	\end{example}
	\begin{example}
		Let $\C_{9,6}$ be a $[9,6,3]_5$ optimal code (from BKLC). There are two matrices
		$$
		G=\left[
		\setlength{\arraycolsep}{3pt}
		\renewcommand{\arraystretch}{0.7}
		\begin{array}{ccccccccc}
			1 & 0 & 0 & 0 & 0 & 0 & 2 & 4 & 4 \\
			0 & 1 & 0 & 0 & 0 & 0 & 4 & 1 & 1 \\
			0 & 0 & 1 & 0 & 0 & 0 & 1 & 0 & 4 \\
			0 & 0 & 0 & 1 & 0 & 0 & 4 & 0 & 2 \\
			0 & 0 & 0 & 0 & 1 & 0 & 2 & 1 & 1 \\
			0 & 0 & 0 & 0 & 0 & 1 & 1 & 3 & 4 \\
		\end{array}\right],
		$$
		and
		$$
		P=\left[
		\setlength{\arraycolsep}{3pt}
		\renewcommand{\arraystretch}{0.7}
		\begin{array}{cccccc}
			0 & 2 & 2 & 0 & 0 & 0 \\
			3 & 1 & 1 & 0 & 0 & 0 \\
			2 & 1 & 0 & 1 & 0 & 0 \\
			0 & 1 & 1 & 1 & 0 & 0 \\
			2 & 4 & 0 & 3 & 0 & 1 \\
			3 & 4 & 1 & 4 & 1 & 0 \\
		\end{array}\right]=[P_1,P_2,\dots,P_6],
		$$
		where $G$ is a generator matrix of $\C_{9,6}$ and $P$ is an invertible matrix such that $$-GG^T=P\text{diag}(1,1,1,2,0,0)P^T .$$
		We have $\dim(\text{Hull}(\C_{9,6}))=6-\text{rank}(GG^T)=2$ and $\C_{9,6}$ is of type \textbf{($\text{E}_{\text{o,ns}}$)}. Note that $2=1^2+1^2$ and $P_4=P(0,0,0,1)^T$. Using Algorithm 2, Table 4 can be depicted as follows.
		\begin{table}[H]
			\centering
			\vspace{-0.2em}
			\renewcommand{\arraystretch}{1.1}
			\begin{tabular}{|c|c|c|c|c|c|}
				\hline
				$t$ & Generator matrix & $[n,k,d]$ &  $t$ & Generator matrix & $[n,k,d]$ \\
				\hline
				$0$ (LCD)&$[G,P_5,P_6]$ &$[11,6,3]$ & $4$&$[G,P_1,P_2]$ &$\textbf{[11,6,4]}$\\
				\hline
				$1$& $[G,P_6]$ & $[10,6,3]$ &  $5$ &$[G,P_1,P_2,P_3]$ &$[12,6,4]$ \\
				\hline
				$3$ &$[G,P_1]$ &$[10,6,3]$ & $6$ (ESO)& $[G,P_1,P_2,P_3,P_4,P_4]$&$\textbf{[14,6,6]}$\\
				\hline
				
			\end{tabular}
			\caption{Shortest $t$-dimensional Euclidean hull embeddings of $\C_{9,6}$}
		\end{table}
		By comparison with the BKLC database, the $[11,6,4]_5$ code generated by $[G,P_1,P_2]$ and the $[14,6,6]_5$ code generated by $[G,P_1,P_2,P_3,P_4,P_4]$ are almost optimal.
	\end{example}
	\begin{algorithm}
		\caption{Construction of a shortest $t$-dimensional Euclidean hull embedding of $q$-ary linear code when $q=2^m$}
		\begin{algorithmic}[1]
			\State	Let $q=2^m$ and let $\C$ be an $[n,k]$ linear code over $\F_{q}$ with generator matrix $G$ and $0 \leq t \leq k$
			\State \textbf{Goal:} Find a $k \times (n+s)$ matrix $\widetilde{G}$ over $\F_{q}$ which generates a shortest $t$-dimensional Euclidean hull embedding of $\C$
			\Procedure{EvenEuclideanHullEmbedding}{$G$, $t$}
			\State $k \gets \text{the number of rows in matrix } G$
			\State $\ell \gets k-\text{rank}(GG^T)$
			\State $\mathcal{A}\gets \{g_{ii}:GG^T=[g_{ij}],1\leq i,j \leq k\}$
			\State  $s \gets  |t-\ell| $
			\State $D \gets  O_{s\times (k-s)} ,J\gets \begin{bmatrix}0 &1 \\ 1& 0\end{bmatrix}$
			\If{$\mathcal{A} \neq \{0\}$}
			\State Compute $k\times k$ invertible matrix $P$ s.t. $GG^T=P\text{diag}(I_{(k-\ell)},0,\dots,0)P^T$
			\If{$t < \ell$}
			\State  $D \gets P[D,I_s]^T$
			\Else
			\State $D \gets  P[I_s,D]^T$
			\EndIf
			\Else
			\State Compute $k\times k$ invertible matrix $P$ s.t. $GG^T=P\text{diag}(J,\dots,J,0,\dots,0)P^T$
			\If{$t < \ell$}
			\State  $D \gets P[D,I_s]^T$
			\ElsIf{ $t-\ell$ is odd}
			\State Compute $s\times s$ invertible matrix $D'$ s.t. $\text{diag}(J,\dots,J,1)=D'D'^T$ ($\frac{s-1}{2}$ $J$-blocks)
			\State $D\gets P\begin{bmatrix}D' &\textbf{0} \\ \textbf{0}& 1 \\O_{(k-s-1)\times s}&\textbf{0}\end{bmatrix}$
			\Else
			\State$P' \gets \begin{bmatrix}1 &1&0 \\ 1& 0 &1\end{bmatrix},r\gets 1$
			\While{$r<s/2$}
			\State $P'\gets \begin{bmatrix}
				P'&O_{2r \times 2}\\
				\textbf{1}_{2,2r+1}&J
			\end{bmatrix}$
			\State $r \gets r+1$
			\EndWhile
			\State $D\gets P\begin{bmatrix} P'\\ O_{(k-s) \times (s+1)}\end{bmatrix}$
			\EndIf
			\EndIf
			\State $\widetilde{G} \gets [G, D]$
			\State \Return $\widetilde{G}$
			\EndProcedure
		\end{algorithmic}
	\end{algorithm}
	\clearpage
	\begin{example}\label{844}
		Let $\mathcal{H}_{2,3}$ be the $[7,4,3]_2$ Hamming code. There are two matrices
		\[G= \left[
		\setlength{\arraycolsep}{4pt}
		\renewcommand{\arraystretch}{0.9}
		\begin{array}{ccccccc}
			1 & 0 & 0 & 0 & 1 & 1 & 0 \\
			0 & 1 & 0 & 0 & 0 & 1 & 1 \\
			0 & 0 & 1 & 0 & 1 & 1 & 1 \\
			0 & 0 & 0 & 1 & 1 & 0 & 1 \\
		\end{array}\right],
		\]
		and
		$$
		P=\left[\setlength{\arraycolsep}{4pt}
		\renewcommand{\arraystretch}{0.9}
		\begin{array}{cccc}
			1 & 0 & 0 & 0 \\
			1 & 0 & 0 & 1 \\
			0 & 1 & 0 & 0 \\
			1 & 0 & 1 & 0
		\end{array}\right]=[P_1,P_2,P_3,P_4],
		$$
		where $G$ is a generator matrix of $\mathcal{H}_{2,3}$ and $P$ is an invertible matrix such that $$GG^T=P\text{diag}(1,0,0,0)P^T .$$ Therefore,  $\dim(\text{Hull}(\mathcal{H}_{2,3}))=4-\text{rank}(GG^T)=3$ and $\mathcal{H}_{2,3}$ is of type \textbf{($\text{E}_{\text{e,na}}$)}. Applying Algorithm 3, we obtain Table 5.
		\begin{table}[H]
			\centering
			\vspace{-0.2em}
			\renewcommand{\arraystretch}{1.1}
			\begin{tabular}{|c|c|c|}
				\hline
				$t$&Generator matrix &$[n,k,d]$ \\
				\hline
				$0$ (LCD)&$[G,P_2,P_3,P_4]$ & $[10,4,3]$\\
				\hline
				$1$&$[G,P_3,P_4]$ &$[9,4,3]$ \\
				\hline
				$2$&$[G,P_4]$ &$[8,4,3]$ \\
				\hline
				$4$ (ESO)&$[G,P_1]$ &$\textbf{[8,4,4]}$ \\
				\hline										
			\end{tabular}
			\caption{Shortest $t$-dimensional Euclidean hull embeddings of $\mathcal{H}_{2,3}$}
		\end{table}
		By comparison with the BKLC database, the $[8,4,4]_2$ linear code generated by $[G,P_1]$ is optimal and is equivalent to BKLC(GF(2),8,4).
	\end{example}
	\begin{example}\label{2068}
		Let $\C_{15,6}$ be a $[15,6,6]_2$ optimal code (from BKLC). There are two matrices
		$$
		G=\left[
		\setlength{\arraycolsep}{3pt}
		\renewcommand{\arraystretch}{0.7}
		\begin{array}{ccccccccccccccc}
			1 & 0 & 0 & 0 & 0 & 0 & 0 & 0 & 1 & 0 & 0 & 1 & 1 & 1 & 1 \\
			0 & 1 & 0 & 0 & 0 & 0 & 1 & 0 & 0 & 0 & 1 & 1 & 0 & 1 & 1 \\
			0 & 0 & 1 & 0 & 0 & 0 & 1 & 1 & 0 & 1 & 1 & 0 & 0 & 0 & 1 \\
			0 & 0 & 0 & 1 & 0 & 0 & 1 & 1 & 1 & 1 & 0 & 0 & 1 & 0 & 0 \\
			0 & 0 & 0 & 0 & 1 & 0 & 0 & 1 & 1 & 1 & 1 & 0 & 0 & 1 & 0 \\
			0 & 0 & 0 & 0 & 0 & 1 & 0 & 0 & 1 & 1 & 1 & 1 & 0 & 0 & 1
		\end{array}\right],
		$$
		and
		$$
		P=\left[
		\setlength{\arraycolsep}{3pt}
		\renewcommand{\arraystretch}{0.7}
		\begin{array}{cccccc}
			1 & 0 & 0 & 0 & 0 & 0 \\
			0 & 1 & 0 & 0 & 0 & 0 \\
			1 & 1 & 1 & 0 & 1 & 0 \\
			1 & 0 & 1 & 0 & 0 & 0 \\
			0 & 0 & 0 & 1 & 1 & 1 \\
			1 & 1 & 1 & 1 & 0 & 0
		\end{array}\right]=[P_1,P_2,\dots,P_6],
		$$
		where $G$ is a generator matrix of $\C_{15,6}$ and $P$ is an invertible matrix such that $$GG^T=P\text{diag}(J,J,0,0)P^T .$$
		We have $\dim(\text{Hull}(\C_{15,6}))=6-\text{rank}(GG^T)=2$ and $\C_{15,6}$ is of type \textbf{($\text{E}_{\text{e,a}}$)}. Applying Algorithm 3, we get Table 6.
		\begin{table}[H]
			\centering
			\vspace{-0.2em}
			\renewcommand{\arraystretch}{1.1}
			\begin{tabular}{|c|c|c|}
				\hline
				$t$ & Generator matrix & $[n,k,d]$ \\
				\hline
				$0$ (LCD)&$[G,P_5,P_6]$ &$[17,6,6]$ \\
				\hline
				$1$& $[G,P_6]$ & $[16,6,6]$ \\
				\hline
				$3$ &$[G,P_1,P_2]$ &$[17,6,6]$ \\
				\hline
				$4$&$[G,P_1+P_2,P_1,P_2]$ &$[18,6,6]$\\
				\hline
				$5$ &$[G,P_2+P_3,P_1+P_3,P_1+P_2+P_3,P_4]$ &$\textbf{[19,6,7]}$ \\
				\hline
				$6$ (ESO)& $[G,P_1+P_2+P_3+P_4,P_1+P_3+P_4,P_2+P_3+P_4,P_4,P_3]$&$\textbf{[20,6,8]}$ \\
				\hline
				
			\end{tabular}
			\caption{Shortest $t$-dimensional Euclidean hull embeddings of $\C_{9,6}$}
		\end{table}
		By comparison with the BKLC database, the $[19,6,7]_2$ code generated by $[G,P_2+P_3,P_1+P_3,P_1+P_2+P_3,P_4]$ is almost optimal; the $[20,6,8]_2$
		code generated by $[G,P_1+P_2+P_3+P_4,P_1+P_3+P_4,P_2+P_3+P_4,P_4,P_3]$ is optimal and is not equivalent to BKLC(GF(2),20,6).
	\end{example}
	\begin{remark}
		The complexity of all the Algorithms 1, 2 and 3 mainly comes from the congruence diagonalization in the steps of ``Compute'' and ``Precompute''. The diagonalization algorithm (see \cite[char. 8]{Wan2004}) uses elementary transformations of matrices to recursively eliminate off-diagonal entries or blocks, leading to a total complexity of at most $O(n^3)$.
	\end{remark}
	\section{Conclusion}
	This paper investigates the problem of extending a generator matrix of an $[n,k]$ linear code with $\ell$-dimensional hull by adding the minimum number of columns, so that the resulting code is an $[n',k]$ linear code with a prescribed hull dimension $t$. For the Hermitian case, we proved that the length of the shortest $t$-dimensional Hermitian hull embeddings is $n+|t-\ell|$. We provided an explicit construction algorithm (Algorithm 1) and applied it to obtain optimal and almost optimal codes.
	
	For the Euclidean case, we classified linear codes into four ``types'' based on the congruence equivalence classes of their Gram matrices. Using the properties of these types, we derived the exact length of the shortest $t$-dimensional Euclidean hull embeddings (Theorem \ref{sel}) and corresponding construction algorithms (Algorithms 2 and 3). In particular, our results for $q=2^m$
	improve the work in \cite{An2025}. Finally, using these algorithms, we constructed several optimal and almost optimal codes from Hamming codes and other known optimal codes, some of which are inequivalent to the BKLC database.

\end{document}